\documentclass[a4paper,10pt]{amsart}
\usepackage[latin1]{inputenc}
\usepackage{a4wide}
\usepackage{amsfonts}
\usepackage{amsmath}
\usepackage{amssymb}
\usepackage{amsthm}

\newtheorem {lemma}{Lemma}
\newtheorem {theorem}{Theorem}
\newtheorem {Corollary}{Condition}
\newtheorem{remark}[theorem]{Remark}

\numberwithin{equation}{section}
\numberwithin{equation}{section}
\numberwithin{equation}{section}
\begin{document}

\author{Detlef D\"urr}
\address{Mathematisches Institut, Ludwig-Maximilians Universit\"at, Theresienstr. 39, 80333 M\"unchen, Germany}
\email{d\"urr@math.lmu.de}
\author{G\"unter Hinrichs}
\address{Mathematisches Institut, Ludwig-Maximilians Universit\"at, Theresienstr. 39, 80333 M\"unchen, Germany}
\email{hinrichs@math.lmu.de}
\author{Martin Kolb}
\address{Department of Statistics, Oxford University, 1 South Parks Road, Oxford, OX1 3TG, United Kingdom}
\email{kolb@stats.ox.ac.uk}

\title[On a stochastic Trotter formula]{On a stochastic Trotter formula with application to spontaneous localization models}
\maketitle
\begin{abstract}
We consider the relation between so called continuous localization models -- i.e. non-linear stochastic Schrödinger evolutions -- and the discrete GRW-model of wave function collapse. The former can be understood as scaling limit of the GRW process.  The proof relies on a stochastic Trotter formula , which is of interest in its own right. Our Trotter formula also allows to complement results on existence theory of stochastic Schr\"odinger evolutions by Holevo and Mora/Rebolledo. 
\end{abstract}
\section{Introduction}
The Ghirardi-Rimini-Weber theory (shortly GRW, see \cite{Ghirardi1}, \cite{BassiG} and, for a mathematically rigorous treatment, \cite{Tumulka})  is one of the best known versions of stochastic localization models, aiming at a solution of the measurement problem in quantum mechanics (see e.g. \cite{Bell}). In the GRW theory the Schr\"odinger equation is interrupted by the random collapse of the wave function at randomly chosen times. In this way, one obtains a description of microscopic quantum phenomena where macroscopic superpositions (i.e. Schrödinger cat states) are surpressed.  The GRW-theory contains two new constants of nature: the collapse of the wave function $\psi$ occurs at random times with intensity $\mu$ and the wave function collapses in space by the multiplication with a Gaussian hitting function with spread $1 / \sqrt{\alpha}$.  The centers of the Gaussians are approximately $|\psi|^2$-distributed.

A related class of stochastic Schr\"odinger evolutions, where the collapse of the wave function happens continuously, appears in the theory of continuous quantum measurement and in the theory of open quantum systems (\cite{Belavkin, Breuer, Joos, Kolokoltsov, Mora}).  The Schrödinger equation is replaced by a Hilbert-space valued stochastic differential equation.

A known example is the Diósi equation (\cite{Diosi1, Diosi2, Penrose, Bassi, Kolb})
\begin{equation}\label{QMUPL-Kollaps}
 \begin{split}
 &d\psi_t(x)=-iH\psi_t\textnormal dt + \sqrt{\lambda}\psi_t(x)x\cdot \textnormal d\xi_t - \frac{\lambda}{2}|x|^2\psi_t(x) \textnormal dt \\
 &\psi_0=\phi_0 ,
 \end{split}
\end{equation}
where $(\xi_t)$ is a three-dimensional Wiener process on a filtered probability space $(\Omega, \mathfrak{F, (F_t)}, \mathbb{Q})$ and $\lambda>0$. The physical wave function process arises from $\psi$ by normalization $\psi/\|\psi\|$ and by a change of measure. The normalized wave function obeys then a non-linear equation.

In the physics literature one often considers the equations for the statistical operators $\rho_t:=\mathbb{E}(|\psi_t\rangle\langle\psi_t|)$. The von-Neumann equation is then replaced by so called Lindblad equations (\cite{Lindblad}), which for the GRW-process reads (\cite{BassiG})
\begin{equation*}
 \frac{d}{dt}\rho_t(x,y) = -\frac{i}{\hbar}[H,\rho_t](x,y) -\mu\left(1-e^{-\frac{\alpha}{4}|x-y|^2}\right)\rho_t(x,y)
 \end{equation*}
and (\cite{Diosi2})
\begin{equation*}
 \frac{d}{dt}\rho_t(x,y)= -\frac{i}{\hbar}[H,\rho_t](x,y) -\frac{\lambda}{2}|x-y|^2\rho_t(x,y)
\end{equation*}
for the Diósi equation. We observe that for $|x-y|^2 << 1/\alpha$ and $\mu \cdot {\alpha}={2\lambda}$ both equations are approximately the same: 
\begin{equation*}
 \mu\left(1-e^{-\frac{\alpha}{4}|x-y|^2}\right)\approx \mu\left(1-(1-\frac{\alpha}{4}|x-y|^2\right) = \frac{\lambda}{2}|x-y|^2 .
\end{equation*}
In \cite{Diosi1}, Diosi addresses therefore the question, whether equation \eqref{QMUPL-Kollaps} can be obtained as a scaling limit from the GRW-process by increasing the spread of the hitting function -- $\alpha \to 0$ -- and, at the same time, the collapse frequency $\mu \to \infty$ , such that 
\begin{equation}\label{Skalierung}
\frac{\mu\cdot \alpha}{2}=\lambda=\text{constant}.
\end{equation} 

We shall prove that this is indeed the case.  A natural method for the proof is suggested by the structure of the GRW-process itself: the Schr\"odinger evolution and the collapse mechanism act alternately. This has obvious similarities with the Trotter formula, which is well known and often used in many areas of mathematical physics. In our situation we need a stochastic version of the Trotter formula, since both the collapse times and collapse centers are random.
Stochastic versions of the Trotter formula have been considered since long (\cite{Kurtz}), but the underlying state space has been assumed to be locally compact, which excludes infinite-dimensional Hilbert spaces. More recently, a version of the Trotter formula for a class of stochastic Schrödinger evolutions with deterministic jump times and bounded collapse operators has been established (\cite{Gough}).

In the following section we introduce the GRW-process and the Diosi process. We observe that the collapse mechanism of GRW can be rephrased in terms of of increments of the Wiener process which appears in the Diosi process.This then allows to rephrase the convergence in terms of a Trotter product formula. In section 3 we first prove the Trotter formula.

\section{The scaling limit of GRW}

In the GRW-process  the Schr\"odinger  evolution is interrupted at the random times
\begin{equation}\label{Sprungzeiten}
T_{n,\mu} :=\sum_{k=1}^n\frac{X_k}{\mu},
\end{equation}
where $X_1, X_2 \dots$ are independent exponentially distributed random variables with $\mathbb E X_k=1$. At the times $T_{n,\mu}$ the wave function gets multiplied by a random Gaussian function with spread $\frac{1}{\sqrt\alpha}$, where the centers of the Gaussians are random variables $Y=(Y_1, Y_2, \dots)$ the distribution of which is given by the wave function in the following way:

We collect the randomness in the sample space $\Omega:=\mathbb R^\mathbb N \times (\mathbb R^3)^\mathbb N$ with coordinate projections $X=(X_1, X_2, \dots)$ on the first and $Y=(Y_1, Y_2, \dots)$ on the second factor. We define $\mathbb P(\tau\in\cdot)$ to be a countable product of exponential distributions, $\mathbb P(Y\in\cdot|\tau)$ will be specified via the wavefunction.
Starting with $\phi_0\in L^2(\mathbb{R}^3,\mathbb{C})$, we define recursively
\begin{equation}\label{GRW}
 \begin{split}
 & \psi_{T_{n,\mu}}^\mu(Y_1, \dots, Y_n, X, x) := \left(\frac{\alpha}{\pi}\right)^{\frac{1}{4}} e^{-\frac{\alpha}{2}|x-Y_n|^2} e^{-\frac{i}{\mu}X_n H} \phi_{T_{n-1}}(Y_{T_{1,\mu}}, \dots, Y_{T_{n-1,\mu}}, X, x) \\
 & \phi_{T_{n,\mu}}^\mu:=\frac{\psi_{T_{n,\mu}}}{\|\psi_{T_{n,\mu}}\|_{L^2(\mathbb{R}^3)}}
 \end{split}
\end{equation}
and
\begin{equation}\label{flash}
\begin{split}
 \mathbb{P}_{\alpha,\mu}(Y_{T_{n,\mu}}\in A\mid Y_{T_{1,\mu}}, \dots, Y_{T_{n-1}}, X)
 :=& \int_A \|\psi_{T_{n,\mu}}^\mu(Y_{T_{1,\mu}},\dots,Y_{T_{n-1,\mu}}, X,y,\cdot)\|_{L^2(\mathbb{R}^3)}^2 \textnormal dy \\
 =& \int_A \sqrt{\frac{\alpha}{\pi}}\int e^{-\alpha|x-y|^2}|(e^{-\frac{i}{\mu}X_n H}\phi_{T_{n-1}})(Y, X,  x)|^2\textnormal{d}x\textnormal{d}y ,
\end{split}
\end{equation}
which reads explicitely as
\begin{equation}\label{flash1}
\begin{split}
 \mathbb{P}_{\alpha,\mu}(Y_{T_{1,\mu}}\in A_1) =& \int_{A_1}\|\psi_{T_{1,\mu}}^\mu(y_1, X, \cdot)\|_{L^2(\mathbb{R}^3)}^2 \textnormal dy_1 \\
 =& \int_{A_1} \sqrt{\frac{\alpha}{\pi}}\int e^{-\alpha|x-y_1|^2}|(e^{-\frac{i}{\mu}X_1 H}\phi_0)(x)|^2\textnormal{d}x\textnormal{d}y_1 ,
\end{split}
\end{equation}

\begin{equation*}
\begin{split}
 &\mathbb{P}_{\alpha,\mu}(Y_{T_{1,\mu}}\in A_1, Y_{T_{2,\mu}} \in A_2 \mid X) \\
 =& \int_{\{Y_{T_{1,\mu}}\in A_1\}} \mathbb{P}_{\alpha,\mu}(Y_{T_{2,\mu}}\in A_2\mid Y_{T_{1,\mu}},X)\textnormal{d}\mathbb{P}_{\alpha,\mu} (\cdot\mid X) \\
 \stackrel{\eqref{flash}}{=}& \int_{\{Y_{T_{1,\mu}}\in A_1\}} \int_{A_2}\sqrt{\frac{\alpha}{\pi}}\int e^{-\alpha|x-y_2|^2}
 |(e^{-\frac{i}{\mu}X_2 H}\frac{\psi_{T_{1,\mu}}^\mu}{\|\psi_{T_{1,\mu}}\|})(Y_{T_{1,\mu}}, X,x)|^2 \textnormal dx\textnormal dy_2 \textnormal d\mathbb{P}_{\alpha,\mu}(\cdot\mid X) \\
 \stackrel{\eqref{flash1}}{=}& \int_{A_1} \int_{A_2}\sqrt{\frac{\alpha}{\pi}}\int e^{-\alpha|x-y_2|^2} |(e^{-\frac{i}{\mu}X_2 H}\psi_{T_{1,\mu}})(y_1, X,x)|^2
 \textnormal dx\textnormal dy_2 \frac{1}{\|\psi_{T_{1,\mu}}(y_1,X,\cdot)\|^2} \|\psi_{T_{1,\mu}}^\mu(y_1,X,\cdot)\|^2 \textnormal dy_1 \\
 \stackrel{\eqref{GRW}}=& \int_{A_1} \int_{A_2}\frac{\alpha}{\pi}\int e^{-\alpha|x-y_2|^2} |(e^{-\frac{i}{\mu}X_2 H}e^{-\frac{\alpha}{2}|\cdot-y_1|^2}e^{-\frac{i}{\mu}X_1 H}\phi_0)(x)|^2
 \textnormal dx\textnormal dy_2 \textnormal dy_1
\end{split}
\end{equation*}
and, inductively,
\begin{equation}\label{flashes}
\begin{split}
 & \mathbb{P}_{\alpha,\mu}(Y_{T_{1,\mu}}\in A_1, \dots, Y_{T_{n,\mu}} \in A_n\mid X) \\
 =& \int_{A_1\times\cdots\times A_n}\sqrt{\frac{\alpha}{\pi}}^n \int |(e^{-\frac{\alpha}{2}|\cdot-y_n|^2}e^{-\frac{i}{\mu}X_n H} \cdots
 e^{-\frac{\alpha}{2}|\cdot-y_1|^2}e^{-\frac{i}{\mu}X_1 H}\phi_0)(x)|^2 \textnormal dx\textnormal d(y_1, \dots, y_n) .
\end{split}
\end{equation}
The sequence of wavefunctions is extended continuously according to
\begin{equation}\label{interpol}
\phi_{T_{n,\mu}+s}^\mu:=e^{-isH}\phi_{T_{n,\mu}}^\mu \textnormal{ for } 0<s<X_{n+1}
\end{equation}
between the collapse times. By this and the preceding equation, the process $(\phi_t^\mu)_{t\ge0}$ on the probability space $(\Omega,\mathfrak F,\mathbb P)$ is specified.

In the Diósi process the wave function obeys the stochastic Schr\"odinger equation
\begin{equation}\label{Diosi}
 \begin{split}
 &d\psi_t(x)=-iH\psi_t\textnormal dt + \sqrt{\lambda}\psi_t(x)x\cdot \textnormal d\xi_t - \frac{\lambda}{2}|x|^2\psi_t(x) \textnormal dt,
 \end{split}
\end{equation}
where $H$ is the Hamiltonian which for the purpose of this paper can be thought of as the Schrödinger Hamiltonian $H=-\frac12\Delta + V$. The noise $(\xi_t)$ is a standard $d$-dimensional Wiener process on a filtered probability space $(\Omega, \mathfrak{F, (F_t)}, \mathbb{Q})$ and $\lambda$ a denotes positive intensity parameter.
 
The ``physical'' collapse process is then given by
\begin{displaymath}
\phi_t:=\frac{\psi_t}{\|\psi_t\|_{L^2(\mathbb{R}^d)}}, 
\end{displaymath}
weighted by a new measure defined as
\begin{equation}\label{cooking}
\mathbb{P}_\lambda(A):= \mathbb{E}_\mathbb{Q}(\chi_A\|\psi_t\|^2) \textnormal{ for } A\in\mathfrak{F}_t .
\end{equation}
The martingale property of $\|\psi_t\|^2$ (see \cite{Mora}) ensures that this is indeed a proper definition of a measure on $\Omega$.

We start the comparison of the GRW and the Diosi processes by first neglecting the Schr\"odinger evolution (i.e. we set $H=0$) and the randomness of the collapse times (i.e. $X_k\equiv1$). Then, from \eqref{GRW} and \eqref{flashes},
\begin{align}\label{871}
 &\psi_{\frac{n}{\mu}}^\mu(x)=\left(\frac\alpha\pi\right)^\frac14 e^{-\frac{\alpha}{2}[|x-Y_1|^2+\cdots+|x-Y_n|^2]}\phi_0(x), \\
 &\phi_{\frac{n}{\mu}}^\mu = \frac{\psi_{\frac{n}{\mu}}}{\|\psi_{\frac{n}{\mu}}\|}
\end{align}
and
\begin{equation}\label{872}
 \mathbb{P}_{\alpha,\mu}(Y_1\in A_1,\dots,Y_n\in A_n)
 = \sqrt{\frac{\alpha}{\pi}}^n \int_{A_1\times\cdots\times A_n} \int e^{-\alpha[(x-y_1)^2+\dots+(x-y_n)^2]}|\phi_0(x)|^2
 \textnormal dx\textnormal d(y_1, \dots, y_n) .
\end{equation}
The solution of the collapse part of \eqref{Diosi} is
\begin{equation}\label{Kollapsloesung}
 \psi_t(x) = e^{\sqrt{\lambda} x\cdot\xi_t - \lambda |x|^2t}\phi_0(x) ,
\end{equation}
and can be rewritten as
\begin{equation}\label{10}
\begin{split}
 \psi_{\frac{n}{\mu}}(x) =& e^{\sum_{k=1}^n[\sqrt{\lambda}x\cdot(\xi_{\frac{k}{\mu}}-\xi_{\frac{k-1}{\mu}}) -
 \frac{\lambda}{\mu}|x|^2]} \phi_0(x)
 = e^{-\frac{\lambda}{\mu}\sum_{k=1}^n [|x-\frac{\mu}{2\sqrt{\lambda}}(\xi_{\frac{k}{\mu}}-\xi_{\frac{k-1}{\mu}})|^2 
 + \frac{\mu^2}{4\lambda}|\xi_{\frac{k}{\mu}}-\xi_{\frac{k-1}{\mu}}|^2]} \phi_0(x)
\end{split}
\end{equation}
or, in normalized form and taking into account the scaling \eqref{Skalierung},
\begin{equation}\label{873}
 \phi_{\frac{n}{\mu}}(x)
 = C e^{-\frac{\alpha}{2}\sum_{k=1}^n |x-\frac{\mu}{2\sqrt{\lambda}}(\xi_{\frac{k}{\mu}}-\xi_{\frac{k-1}{\mu}})|^2} \phi_0(x) .
\end{equation}
The increments
\begin{equation}\label{incr}
Z_k:=\frac{\mu}{2\sqrt{\lambda}}(\xi_{\frac{k}{\mu}}-\xi_{\frac{k-1}{\mu}})
\end{equation}
formally correspond to the $Y_k$ from GRW and, by the aid of \eqref{cooking}, their distribution under $\mathbb P_\lambda$ can be determined if one takes into account that they have independent centered normal distributions with $\mathbb{V}(Z_k)=\frac{\mu}{4\lambda}=\frac{1}{2\alpha}$ under $\mathbb{Q}$:
\begin{equation}\label{874}
 \begin{split}
  &\mathbb{P}_\lambda(Z_1\in A_1, \dots, Z_n\in A_n)
  \stackrel{\eqref{cooking}}{=} \mathbb{E}_\mathbb{Q}(\chi_{\{Z_1\in A_1, \dots, Z_n\in A_n\}}\|\psi_\frac{n}{\mu}\|^2) \\
  \stackrel{\eqref{10}}=& \mathbb{E}_\mathbb{Q}(\chi_{\{Z_1\in A_1, \dots, Z_n\in A_n\}} \int
  e^{\sum_{k=1}^n[2\sqrt{\lambda}x\cdot(\xi_{\frac{k}{\mu}}-\xi_{\frac{k-1}{\mu}}) -
  \frac{2\lambda}{\mu}|x|^2]} |\phi_0|^2(x) \textnormal dx ) \\
  \stackrel{\eqref{incr}}=& \mathbb{E}_\mathbb{Q}(\chi_{\{Z_1\in A_1, \dots, Z_n\in A_n\}} \int
  e^{\sum_{k=1}^n[\frac{4\lambda}{\mu}x\cdot Z_k -
  \frac{2\lambda}{\mu}|x|^2]} |\phi_0|^2(x) \textnormal dx ) \\
  \stackrel{\eqref{Skalierung}}=& \int_{A_1\times\dots\times A_n} \int e^{\sum_{k=1}^n[2\alpha x\cdot z_k - \alpha|x|^2]} |\phi_0|^2(x) \textnormal dx
  \sqrt{\frac{\alpha}{\pi}}^n e^{-\alpha\sum_{k=1}^n |z_k|^2} \textnormal d(z_1,\dots,z_n) \\
  =& \sqrt{\frac{\alpha}{\pi}}^n \int_{A_1\times\dots\times A_n} \int e^{-\alpha\sum_{k=1}^n|x-z_k|^2} |\phi_0|^2(x) \textnormal dx
  \textnormal d(z_1,\dots,z_n)
 \end{split}
\end{equation}
\eqref{873} and \eqref{874} agree with \eqref{871} and \eqref{872}. One sees that the GRW collapse process with deterministic times with parameters $(\frac{2\lambda}{\mu},\mu)$ can be obtained by restricting Diosi's collapse process with parameter $\lambda$ to the appropriate discrete instants of time.

We show now that the full GRW-process given by the parameters $\alpha$ and $\mu$ converges in the sense of finite dimensional distributions to the Diosi process with parameter $\lambda$, when $\mu \to \infty$ and $\mu \alpha/2 = \lambda= \text(constant)$.
\begin{theorem}\label{thm:Limes}
Let $H=-\frac12\triangle + V$ with a bounded potential V having bounded first and second derivatives. Then, for all $t_1, \dots, t_n \ge0$,
$$\lim_{\mu \to \infty} \mathbb{P}_{\frac{2\lambda}{\mu},\mu}\circ(\phi_{t_1}^\mu,\dots,\phi_{t_n}^\mu)^{-1}=\mathbb{P}_{\lambda}\circ(\phi_{t_1},\dots,\phi_{t_n})^{-1}.$$
\end{theorem}
\begin{remark}
 The restrictions on V make sure that the Hamiltonian satisfies conditions (I) and (II) (cf. Condition \ref{Bed}) from the next chapter. It exemplifies that our abstract result can be applied to cases of physical interest. However our focus  in this paper is not on greatest possible generality.
\end{remark}

\begin{proof}

For simplicity, we formulate the proof in one-dimensional space. The key ingredient is the Trotter formula \eqref{produkt}, which is shown in the next section.

Similar to the observation above, we shall express the full GRW process in terms of the increments of the Wiener process by defining
\begin{equation}\label{480}
 \tilde\psi_t^\mu:= e^{-i(t-T_{\kappa_\mu(t)})H}\prod_{k=0}^{\kappa_\mu(t)-1} e^{\sqrt\lambda (\cdot)\cdot(\xi_\frac{k+1}{\mu}-\xi_\frac{k}{\mu})-\frac{\lambda}{\mu}|\cdot|^2}
 e^{-\frac{i}{\mu}X_{k+1}H}\phi_0
\end{equation}
with
\begin{equation}\label{kappa1}
 \kappa_\mu(t):=max\{ k:\sum_{j=1}^k \frac{X_j}\mu\le t\} ,
\end{equation}
setting
\begin{equation*}
 \tilde\phi_t^\mu:=\frac{\tilde\psi_t^\mu}{\|\tilde\psi_t^\mu\|_{L^2}}
\end{equation*}
and performing the reweighting
\begin{equation}\label{cooking2}
 \mathbb{P}_\lambda^\mu(A\mid X):=\mathbb{E}_\mathbb{Q}(\chi_A\|\tilde\psi_{T_{n,\mu}}^\mu\|^2\mid X) \textnormal{ for } A\in\mathfrak{F}_{\frac n\mu} .
\end{equation}
We shall now show that this indeed yields the GRW process, i.e.
\begin{equation}\label{315}
 \tilde\phi_t^\mu \stackrel{\mathcal L}= \phi_t^\mu
\end{equation}
where $\tilde\phi_t^\mu$ is distributed w.r.t. $\mathbb P_\lambda^\mu$ and $\phi_t^\mu$ w.r.t. $\mathbb P_{\alpha,\mu}$.
In the next section we prove the product formula \eqref{produkt}. Observing that \eqref{480} is the right-hand side of \eqref{produkt} (with $H=-\frac12\triangle + V$ and $A\psi(x)=\sqrt\lambda x\psi(x)$), the proof of the theorem reduces to showing that
\begin{enumerate}
 \item[\bf Step 1:] \eqref{315} holds true
 \item[\bf Step 2:]  $-\frac 12\triangle + V$ and the position operator  satisfy Condition \ref{Bed} of Theorem \ref{thm:Produkt}.
 \item[\bf Step 3:] the resulting $L^2$-convergence $\tilde\psi_t^\mu \to \psi_t$ implies $\tilde\phi_t^\mu\to\phi_t$ in the sense of finite-dimensional distributions.
\end{enumerate}

\textbf{Step 1:} Substituting \eqref{incr} into \eqref{480} as we did in \eqref{873}, we get
\begin{equation*}
 \tilde\phi_t^\mu=\frac{\tilde\psi_t^\mu}{\|\tilde\psi_t^\mu\|_{L^2}}=
C e^{-i(t-T_{\kappa_\mu(t)})H}\prod_{k=0}^{\kappa_\mu(t)-1} e^{-\frac{\alpha}{2}|\cdot-Z_{k+1}|^2}e^{-\frac{i}{\mu}X_{k+1}H}\phi
\end{equation*}
with $\xi$-dependent normalization factor C. Observing that, by definition \eqref{incr}, \\
$\{Z_1\in A_1,\dots, Z_n\in A_n\} \in \mathfrak F_\frac{n}{\mu}$, we arrive at the following generalization of \eqref{874} :
\begin{align*}
  &\mathbb{P}_\lambda^\mu(Z_1\in A_1,\dots, Z_n\in A_n\mid X)
  \stackrel{\eqref{cooking2}}= \mathbb{E}_\mathbb{Q}(\chi_{\{Z_1\in A_1, \dots, Z_n\in A_n\}}\|\tilde\psi_{T_{n,\mu}}^\mu\|^2\mid X) \\
  \stackrel{\eqref{480}, \eqref{Sprungzeiten}}=& \mathbb{E}_\mathbb{Q}(\chi_{\{Z_1\in A_1, \dots, Z_n\in A_n\}} \int |\prod_{k=0}^{n-1}
  e^{\frac{2\lambda}{\mu}(\cdot)\cdot Z_{k+1} - \frac{\lambda}{\mu}|\cdot|^2}e^{-\frac{i}{\mu}X_{k+1}H}\phi|^2(x) \textnormal dx \mid X ) \\
  \stackrel{\eqref{Skalierung}}=& \int_{A_1\times\dots\times A_n} \int |\prod_{k=0}^{n-1}
  e^{\alpha (\cdot)\cdot z_{k+1} - \frac{\alpha}{2}|\cdot|^2}e^{-\frac{i}{\mu}X_{k+1}H}\phi|^2(x) \textnormal dx
  \sqrt{\frac{\alpha}{\pi}}^n e^{-\alpha\sum_{k=0}^{n-1} |z_{k+1}|^2} \textnormal d(z_1,\dots,z_n) \\
  =& \sqrt{\frac{\alpha}{\pi}}^n\int_{A_1\times\dots\times A_n} \int |\prod_{k=0}^{n-1}
  e^{\alpha (\cdot)\cdot z_{k+1} - \frac{\alpha}{2}|\cdot|^2 - \frac{\alpha}{2}|z_{k+1}|^2}e^{-\frac{i}{\mu}X_{k+1}H}\phi|^2(x) \textnormal dx
  \textnormal d(z_1,\dots,z_n) \\
  =& \sqrt{\frac{\alpha}{\pi}}^n\int_{A_1\times\dots\times A_n} \int |\prod_{k=0}^{n-1}
  e^{-\frac{\alpha}{2}|\cdot-z_{k+1}|^2}e^{-\frac{i}{\mu}X_{k+1}H}\phi|^2(x) \textnormal dx \textnormal d(z_1,\dots,z_n) .
\end{align*}
Comparing these equations with \eqref{GRW}, \eqref{interpol} and \eqref{flashes} - the ones defining GRW -, one sees that $\mathbb P_\lambda^\mu = \mathbb P_{\alpha, \mu}$ so that, in particular, \eqref{315} holds true.

\textbf{Step 2:}  
Let $\mathcal M$ be the domain of the harmonic oscillator. According to \cite{Huang}, this is invariant under $H_t$ and, since $A_{s,t}$ is just multiplication by a Gaussian (\eqref{Kollapsloesung}, \eqref{e:partialflowA}), also under $A_{s,t}$.
Next, we choose $\mathcal N:=C_0^\infty$ and verify the second part of  Condition \ref{Bed} (I)) by an explicit calculation, starting with the case $V\equiv0$:
\begin{align*}
 &\mathbb E \int |\Delta\left(e^{\sqrt\lambda x\xi_t-\lambda x^2t}-1\right)\phi(x)|^2 \textnormal dx \\
 \le& 3\mathbb E\int |\left((\xi_t-2tx)^2-2t\right)e^{\sqrt\lambda x\xi_t-\lambda x^2t}|^2|\pi(x)|^2 \textnormal dx
 + 12\mathbb E\int|(\xi_t-2tx)e^{\sqrt\lambda x\xi_t-\lambda x^2t}|^2|\phi'(x)|^2 \\
 &+ 3\mathbb E\int |e^{\sqrt\lambda x\xi_t-\lambda x^2t}-1|^2|\psi''(x)|^2\textnormal dx \\
 =& 3\int\frac1{\sqrt{2\pi t}}\int \left((y-2tx)^2-2t\right)^2e^{-\frac{(y-2tx)^2}{2t}}\textnormal dy|\phi(x)|^2 \textnormal dx \\
 &+ 12\int\frac1{\sqrt{2\pi t}}\int (y-2tx)^2 e^{-\frac{(y-2tx)^2}{2t}}\textnormal dy|\phi'(x)|^2 \textnormal dx \\
 &+ 3\int\frac1{\sqrt{2\pi t}}\int (e^{\sqrt\lambda xy-tx^2}-1)^2 e^{-\frac{y^2}{2t}}\textnormal dy|\phi''(x)|^2 \textnormal dx \\
 =& 15t^2\int|\phi(x)|^2 \textnormal dx
 + 12t\int |\phi'(x)|^2 \textnormal dx
 + 6\int (1-e^{-\frac t2x^2})|\phi''(x)|^2 \textnormal dx
\end{align*}
The first two summands trivially converge to 0 for $t\to0$.  The last one converges to zero by dominated convergence because the integrand is dominated by $\frac t2x^2|\phi''(x)|^2$. 

If $V\neq0$ and bounded by a constant $C$, then by virtue of \eqref{Anfangswert} 
\begin{equation*}
 \mathbb E\|(\Delta+V)(A_{0,t}-1)\phi\|^2 
\le 2 \mathbb E\|\Delta(A_{0,t}-1)\phi\|^2 + 2C \mathbb E\|(A_{0,t}-1)\phi\|^2 \to0 \,.
\end{equation*}

Condition \ref{Bed} (II) was verified by \cite{Mora} for the case $V\equiv0$. $R$ was again chosen to be the harmonic oscillator.  $V\neq0$ bounded does not change the domains, so it remains to establish
\begin{equation*}
 \langle R^2\psi, V\psi\rangle \le \alpha(\|\psi\|^2 + \|R\psi\|^2 + \beta)
\end{equation*}
for some $\alpha, \beta>0$. Indeed,
\begin{equation*}
 \langle R^2\psi, V\psi\rangle = \langle R\psi, RV\psi\rangle \le \|R\psi\|\cdot\|RV\psi\| \le \frac12 \|R\psi\|^2 + \frac12 \|RV\psi\|^2 \le \alpha(\|\psi\|^2 + \|R\psi\|^2 + \beta) ,
\end{equation*}
the last inequality being ensured by the boundedness of $V$ and its derivatives.

\textbf{Step 3:} We have to show that
\begin{equation*}
 \lim_{\mu\to\infty}\mathbb{E}_{\mathbb P_\lambda^\mu}(f(\tilde\phi^\mu_{t_1},\dots,\tilde\phi^\mu_{t_n}))
 = \mathbb{E}_{\mathbb P_\lambda}(f(\phi_{t_1},\dots,\phi_{t_n}))
\end{equation*}
for all bounded continuous $f: L^2(\mathbb R)^n\to \mathbb R$.
Observing that, by definition \eqref{kappa1}, there are $\kappa_\mu(t_n)$ jumps up to time $t_n$, which implies that, for fixed X,
 $f(\tilde\phi^\mu_{t_1},\dots,\tilde\phi^\mu_{t_n})$ is $\mathfrak F_\frac{\kappa_\mu(t_n)}{n}$-measurable, and that, by \eqref{Sprungzeiten},
$$T_{\kappa_\mu(t_n),\mu}\le t_n ,$$
\eqref{cooking2} yields
\begin{equation*}
 \mathbb{E}_{\mathbb P_\lambda^\mu}(f(\tilde\phi^\mu_{t_1},\dots,\tilde\phi^\mu_{t_n}))
=\mathbb{E}_\mathbb{Q}(f(\tilde\phi^\mu_{t_1},\dots,\tilde\phi^\mu_{t_n}) \|\tilde\psi^\mu_{t_n}\|^2 ) .
\end{equation*}
Taking into account \eqref{cooking}, we are thus left to show
\begin{equation}\label{a10}
 \lim_{\mu\to\infty} \mathbb{E}_\mathbb{Q}(f(\tilde\phi^\mu_{t_1},\dots,\tilde\phi^\mu_{t_n}) \|\tilde\psi^\mu_{t_n}\|^2 )
= \mathbb{E}_\mathbb{Q}(f(\phi_{t_1},\dots,\phi_{t_n}) \|\psi_{t_n}\|^2 ) .
\end{equation}
Now, by triangle inequality,
\begin{equation*}\label{a11}
\begin{split}
 &|f(\tilde\phi^\mu_{t_1},\dots,\tilde\phi^\mu_{t_n}) \|\tilde\psi^\mu_{t_n}\|^2
- f(\phi_{t_1},\dots,\phi_{t_n}) \|\psi_{t_n}\|^2| \\
\le& |f(\tilde\phi^\mu_{t_1},\dots,\tilde\phi^\mu_{t_n})| \cdot |\|\tilde\psi^\mu_{t_n}\|^2 - \|\psi_{t_n}\|^2|
+ |f(\tilde\phi^\mu_{t_1},\dots,\tilde\phi^\mu_{t_n}) - f(\phi_{t_1},\dots,\phi_{t_n}) | \|\psi_{t_n}\|^2
\end{split}
\end{equation*}
The first summand goes to 0 in $L^1(\Omega)$ because f is bounded and $\tilde\psi^\mu_{t_n} \to \psi_{t_n}$ in $L^2$ according to { Step 2}. The second summand is dominated by $4C^2\|\psi_{t_n}\|^2$ where C is a bound for f, so, according to Vitali's theorem, \eqref{a10} is proven once we know that
\begin{equation}\label{a12}
 \lim_{\mu\to\infty} |f(\tilde\phi^\mu_{t_1},\dots,\tilde\phi^\mu_{t_n}) - f(\phi_{t_1},\dots,\phi_{t_n}) | \|\psi_{t_n}\|^2 = 0
\end{equation}
in probability w.r.t. $\mathbb Q$. For this, let $\epsilon > 0$. From \cite{Mora2}, $\|\psi_t\|>0$ almost surely for every t, so, by $\sigma$-continuity, one can find an $a$ such that
\begin{equation}\label{ex1}
 \mathbb Q(\|\psi_{t_1}\| , \dots, \|\psi_{t_n}\| >a) > 1-\frac\epsilon3 .
\end{equation}
Since $\|\tilde\psi^\mu_{t_n}\| \to \|\psi_{t_n}\|$ in $L^2$, this result also allows us to find $M\in\mathbb N$ and $b>0$ such that, for all $\mu>M$,
\begin{equation}\label{ex2}
 \mathbb Q(\|\tilde\psi^\mu_{t_1}\|, \dots, \|\tilde\psi^\mu_{t_n}\|>b) > 1-\frac\epsilon3 .
\end{equation}
Next, we choose $\delta>0$ such that, if
$$\|\chi_1 - \phi_{t_1}\| < \delta, \dots, \|\chi_n - \phi_{t_n}\| < \delta ,$$
then
$$|f(\chi_1, \dots, \chi_n) - f(\phi_{t_1}, \dots, \phi_{t_n})| < \epsilon .$$
Finally, again applying the $L^2$-convergence $\tilde\psi^\mu_t \to \psi_t$, we fix $N\in\mathbb N$ such that for all $\mu>N$
\begin{equation}\label{ex3}
 \mathbb Q(\|\tilde\psi^\mu_{t_1}-\psi_{t_1}\|, \dots, \|\tilde\psi^\mu_{t_1}-\psi_{t_1}\|<\delta\min(a,b)) > 1-\frac\epsilon3 .
\end{equation}
Applying the estimate
\begin{equation*}
 \left\| \frac{x}{\|x\|} - \frac{y}{\|y\|} \right\| \le \frac{\|x-y\|}{\min(\|x\|,\|y\|)},
\end{equation*}
which holds true on any Hilbert space, to $x=\tilde\psi^\mu_{t_k}$ and $y=\psi_{t_k}$, we find that
$$\|\tilde\phi^\mu_{t_1}-\phi_{t_1}\|, \dots, \|\tilde\phi^\mu_{t_n}-\phi_{t_n}\| <\delta$$
outside the exceptional sets from \eqref{ex1}, \eqref{ex2} and \eqref{ex3}, so
\begin{equation*}
 \mathbb Q( |f(\tilde\phi^\mu_{t_1},\dots,\tilde\phi^\mu_{t_n}) - f(\phi_{t_1},\dots,\phi_{t_n})|>\epsilon ) <\epsilon .
\end{equation*}

\end{proof}
\section{Stochastic Trotter formula}
Let $H$ and $A$ be possibly unbounded selfadjoint operators in the separable Hilbert space $\mathfrak{H}$ with domains $\mathcal{D}(H)$ and $\mathcal{D}(A)$, respectively.  
In the Diosi process $H=-\frac12\Delta $ and $A$ is multiplication by $\sqrt\lambda x$.
Consider the equation 
\begin{equation}\label{e:SSE}
\begin{split}
   \textnormal{d}\psi_t =& (-iH-\frac{1}{2}A^2)\psi_t \textnormal{d}t + A\psi_t\textnormal{d}\xi_t \\
   \psi_0=&\psi
\end{split}
\end{equation}
and call its solution flow $C_t$.
The operators $H$ and $A$ generate the partial solution flows
\begin{equation}\label{e:partialflowH}
H_t:=e^{itH}
\end{equation}
and
\begin{equation}\label{e:partialflowA}
A_{s,t}:=e^{(\xi_t-\xi_s)A-(t-s)A^2}.
\end{equation}

It is natural to consider the whole ``anti-hermitian'' part of \eqref{e:SSE}, i.e.
\begin{equation}\label{921}
\textnormal{d}\psi_t = A\psi_t\textnormal{d}\xi_t - \frac{1}{2}A^2\psi_t\textnormal{d}t ,
\end{equation}
as the `stochastic part'. Due to the damping term $-\frac12 A^2$, its solution \eqref{e:partialflowA}
is in $L^2(\Omega,\mathfrak H)$, depends continuously in norm on the initial data, i.e.
\begin{equation}\label{Anfangswert}
 \lim_{t \to 0}\mathbb E\|(A_{0,t}-1)\psi\|^2 = 0
\end{equation}
and $\|\psi_t\|^2$ is a martingale, in particular
\begin{equation}\label{iso1}
 \mathbb{E}\|\psi_t\|^2=\|\psi_0\|^2 .
\end{equation}
These useful properties can easily be verified via the functional calculus; the martingale property and \eqref{iso1} also hold true for strong solutions of \eqref{e:SSE} (see \cite{Mora}). If one increased the factor $\frac{1}{2}$ in the damping term, \eqref{921} would still have a strong solution with $\|\psi_t\|^2$ being a submartingale, Holevo's weak solution (\cite{Holevo}) to the correspondingly modified equation \eqref{e:SSE} would have the same property, but Mora's result about strong solvability (\cite{Mora}) would not apply. Without any damping term, the solution would in general not be in $L^2(\Omega, \mathfrak H)$. Of course, if one lets A be, for example, the position operator, it could still be defined pointwise.

We assume throughout this section
\begin{Corollary}\label{Bed}
\begin{enumerate}
\item[]
\item[(I)] There exists a dense subset $\mathcal{M} \subset Dom(H)\cap Dom(A^2)$ of $\mathfrak H$ which is invariant under $H_t$ and $A_{s,t}$, and another dense set $\mathcal{N}\subset Dom(H^2)\cap Dom(A^4)$ such that for all $\phi \in \mathcal N$
$$\lim_{t\to0} \mathbb E\|H(A_{0,t}-1)\phi\|^2 =0 .$$
\item[(II)] There exists a self-adjoint positive ``reference operator'' $R$ such that, with $$G:=-iH-\frac12 A^2 ,$$
	\begin{itemize}
		\item $Dom(R)\subset Dom(G)\cap Dom(G^*)$, $Dom(R^2)\subset Dom(RA)$
		\item $\left\{ \psi\in\mathfrak H \mid \|\psi\|^2 + \|R\psi\|^2 \le1 \right\}$ is compact
		\item $\langle R^2\psi, G\psi\rangle + \|A\psi\|^2 \le \alpha(\|\psi\|^2 + \|R\psi\|^2 + \beta)$ for some $\alpha,\beta \ge 0$ .
	\end{itemize}
\end{enumerate}
\end{Corollary}
(II) is taken from \cite{Mora} and guarantees that \eqref{e:SSE} has a so-called ``$R$-strong solution'' for any regular initial value.

The following theorem is a version of the Trotter formula for the stochastic Schr\"odinger equation.
\begin{theorem}\label{thm:Produkt}
Let $X_1, X_2, \dots$ be independent exponential random variables, which are independent of the Wiener process $\xi$. Define
\begin{equation}\label{kappa}
 \kappa_\mu(t) := \max\{k\in\mathbb{N}:\sum_{j=1}^k \frac{X_i}{\mu}\le t\} ,
\end{equation}
i.e. the corresponding Poisson process with intensity $\mu$, and
\begin{equation*}\label{Sprungzeiten2}
T_{n,\mu} :=\sum_{k=1}^n\frac{X_k}{\mu},
\end{equation*}
the time of the n'th jump.

Moreover, assume that Condition \ref{Bed}(I) is satisfied. Then the following formula holds true in the weak $L^2$-sense:
\begin{equation}\label{produkt}
  C_{0,t}\psi=\lim_{\mu\to\infty} H_{t-T_{\kappa_\mu(t),\mu}}\prod_{k=0}^{\kappa_\mu(t)-1} A_{\frac{k}{\mu},\frac{k+1}{\mu}} H_{\frac{X_{k+1}}\mu} \psi 
 \end{equation}
If, in addition, (II) is satisfied and $\mathcal M\cap Dom(R)$ is dense in $\mathfrak{H}$,  the convergence is  strong.
\end{theorem}
\begin{remark}
 Formula \eqref{produkt} should be read in the following way: A time-continuous evolution, given by $H_t$ with time parameter t, is interrupted at random times by jumps, described by $A_{\frac k\mu, \frac{k+1}{\mu}}$. The jumps are Markovian, there sizes are independent and do not depend on the waiting times. In the proof, the indices of $A_{\frac k\mu, \frac{k+1}{\mu}}$ will be treated on the same level as the time parameter; however, this is only a technical trick.
\end{remark}

\begin{proof}
 W.l.o.g. we first construct the Wiener process and the i.i.d. sequence on separate probability spaces $(\Omega_\xi, \mathfrak F_\xi,\mathbb P_\xi)$ and $(\Omega_X, \mathfrak F_X,\mathbb P_X)$, then form the product space
\begin{equation*}
(\Omega, \mathfrak F,\mathbb P) := (\Omega_\xi\times\Omega_X, \mathfrak F_\xi\otimes\mathfrak F_X,\mathbb P_\xi\otimes \mathbb P_X)
\end{equation*}
and get $\xi$ and $X$ as the appropriate projections.

We define the shorthand notation.
\begin{equation}\label{b1}
 g_\mu(t):= \prod_{k=0}^{\kappa_\mu(t)-1} H_{\frac{X_{k+1}}\mu} A_{\frac{k}{\mu},\frac{k+1}{\mu}} .
\end{equation}
Our key idea is to decompose the total error $g_\mu(t)-C_{0,t}$ into the sum over the local errors over time intervals $\left[\frac k\mu,\frac{k+1}\mu\right]$. Therefore we rewrite $g_\mu(t)$ as a telescopic sum:
\begin{equation}\label{b2}
 g_\mu(t) = 1+\sum_{k=0}^{\kappa_\mu(t)-1}
 \left(A_{\frac{k}{\mu},\frac{k+1}{\mu}}H_{\frac{X_{k+1}}\mu}-1\right) g_\mu\left( T_{k,\mu}\right) .
\end{equation}

 If we choose an initial value $\psi\in\mathcal M$, $u\ge0$ and $\mu\in\mathbb N$, then Condition \ref{Bed}(I) ensures that for fixed $s\ge 0$ and variable $t\ge s$ one can apply the It\^o formula to $A_{s,t}H_{c(t-s)}g_\mu(u)\psi$. We get
 \begin{equation*}\label{bito}
 \begin{split}
  A_{s,t}H_{c(t-s)}g_\mu(u)\psi = g_\mu(u)\psi-&i\int_s^t cA_{s,X}H_{c(X-s)}Hg_\mu(u)\psi\mathrm d\tau + \int_s^t A_{s,\tau}A H_{c(\tau-s)}g_\mu(u)\psi\mathrm d\xi_\tau \\
  -& \frac{1}{2}\int_s^t A_{s,\tau}A^2 H_{c(\tau-s)} g_\mu(u)\psi\mathrm d\tau .
 \end{split}
 \end{equation*}
The formula remains valid for $c=\frac{X_k}\mu$, $s=k$ and $t=k+1$. Substituting it into \eqref{b2}, we get
 \begin{equation}\label{b3}
 \begin{split}
  g_\mu(t)\psi
 = \psi-&i\sum_{k=0}^{\kappa_\mu(t)-1} \int_\frac{k}{\mu}^{\frac{k+1}{\mu}} X_{k+1} A_{\frac{k}{\mu},s} H_{X_{k+1}(s-\frac{k}{\mu})}H g_\mu\left( T_{k,\mu}\right)\psi \mathrm ds \\
 +& \sum_{k=0}^{\kappa_\mu(t)-1} \int_{\frac{k}{\mu}}^{\frac{k+1}{\mu}} A_{\frac{k}{\mu},s}A H_{X_{k+1}(s-\frac{k}{\mu})} g_\mu\left( T_{k,\mu}\right)\psi \mathrm d\xi_s \\
 -&\frac{1}{2} \sum_{k=0}^{\kappa_\mu(t)-1} \int_{\frac{k}{\mu}}^{\frac{k+1}{\mu}} A_{\frac{k}{\mu},s}A^2 H_{X_{k+1}(s-\frac{k}{\mu})} g_\mu\left( T_{k,\mu}\right)\psi \mathrm ds \\
 =: \psi +& I_1^\mu(t)+I_2^\mu(t)+I_3^\mu(t) .
 \end{split}
 \end{equation}
We establish the weak limit of the above expressions. Note that  the operators in front of the $g_\mu$'s will converge to the operators $H,A\textnormal{ and }A^2$ appearing in \eqref{e:SSE} for $n\to\infty$.

Our first step is to show equicontinuity in a weak sense in $L^2(\Omega,\mathfrak{H})\cong L^2(\Omega_\xi,\mathbb{C})\otimes L^2(\Omega_X,\mathbb{C})\otimes\mathfrak{H}$, i.e. we choose $\zeta\in L^2(\Omega_\xi,\mathbb{C})$, $Y\in L^2(\Omega_X,\mathbb{C})$ and $\phi\in\mathcal{N}$, consider the complex-valued functions
\begin{equation}\label{b4}
f_\mu(t):=\mathbb{E}\overline{ Y\zeta}\langle\phi,g_\mu(t)\psi\rangle
\end{equation}
and show that they are equicontinuous as a sequence in $\mu$. Then we will conclude convergence for $\mu\to\infty$ and argue that the limit actually solves \eqref{e:SSE}.

Since $H_t$ and $A_{s,t}$ preserve the $L^2(\Omega_\xi,\mathfrak H)$-norm (see \eqref{iso1}),
\begin{equation}\label{b5}
\mathbb{E}\|g_\mu(t)\psi\|^2=\mathbb E_X\mathbb E_\xi\|g_\mu(t)\psi\|^2=\mathbb E_X\|\psi\|^2=\|\psi\|^2
\end{equation}
holds, so the terms $g_\mu(\cdot)\psi$, viewed as mappings from some bounded time interval $[0,T]$ to $L^2(\Omega,\mathfrak H)$, are uniformly bounded for any $\psi\in\mathfrak H$.
Therefore, in order to get weak convergence of $g_\mu(\cdot)\psi$ for $\mu\to\infty$, we need only consider $\zeta$, $Y$ and $\phi$ from dese subsets of the corresponding Hilbert spaces; in particular, we restrict to
\begin{equation*}
\phi\in\mathcal{N} .
\end{equation*}
Since $L^2(\Omega,\mathfrak H)$ is separable ($L^2(\Omega_\xi)$ according to the Wiener chaos decomposition (\cite{Hackenbroch}), $L^2(\Omega_X,\mathbb C)$ by construction, $\mathfrak H$ by assumption and the whole space $L^2{\Omega, \mathfrak H}$ due to its tensor structure), even an appropriate sequence
\begin{equation}\label{b5.5}
(Y_l\zeta_l\phi_l)_{l\in\mathbb N}
\end{equation}
suffices; for simplicity of notation, we can take normalized vectors.

In the following, we shall repeatedly apply the inequality
\begin{equation}\label{Guenter}
 \|x_1 + \dots + x_n\|^2 \le n(\|x_1\|^2 + \dots + \|x_n\|^2) , 
\end{equation}
which holds true for any $x_1, \dots, x_n$ from a vector space with norm $\|\cdot\|$ and $n\in\mathbb N$, to the norms on $\mathbb H$ , $L^2(\Omega, \mathbb C)$ and $L^2(\Omega, \mathfrak H)$. Moreover, we shall often employ  the integral form of Jensen's inequality
\begin{equation}\label{Jensen}
 \left\| \int_a^b \dots \textnormal dx \right\|^2 \le |b-a|\int_a^b \| \dots \|^2 \textnormal dx\,.
\end{equation}

According to the Cauchy-Schwarz inequality and \eqref{b5}, every sequence $(f_\mu)$ is bounded uniformly in t. As for the equicontinuity, we recall definition \eqref{b4} and formula \eqref{b3} and, using the Cauchy-Schwarz inequality, we compute
\begin{equation}\label{b6}
 |f_\mu(t)-f_\mu(s)|^2 \le \mathbb{E}|Y\zeta|^2\cdot\mathbb{E}|\langle\phi,g_\mu(t)\psi-g_\mu(s)\psi\rangle|^2
 \stackrel{\eqref{Guenter}}\le 3\sum_{j=1}^3\mathbb{E}|\langle\phi,I_j^\mu(t)-I_j^\mu(s)\rangle|^2 .
\end{equation}

Now we will show equicontinuity for the three summands in \eqref{b6} separately. Recalling their definition \eqref{b3}, we start with
\begin{equation}\label{b7}
 \begin{split}
 &\mathbb{E}|\langle\phi,I_1^\mu(t)-I_1^\mu(s)\rangle|^2 \\
 =&\mathbb{E}|\langle\phi, \sum_{k=\kappa_\mu(s)}^{\kappa_\mu(t)-1} \int_{\frac{k}\mu}^{\frac{k+1}\mu} X_{k+1}  A_{\frac{k}\mu,s} H_{X_{k+1}(s-\frac{k}\mu)}H g_\mu\left( T_{k,\mu}\right)\psi \mathrm ds\rangle|^2 \\
 =&\mathbb{E}_X\mathbb E_\xi|\sum_{k=\kappa_\mu(s)}^{\kappa_\mu(t)-1} X_{k+1} \int_{\frac{k}\mu}^{\frac{k+1}\mu} \langle  H_{X_{k+1}(s-\frac{k}\mu)}H A_{\frac{k}\mu,s} \phi,  g_\mu\left( T_{k,\mu}\right)\psi \rangle \mathrm ds|^2 \\
 \stackrel{\eqref{Guenter}}\le& \mathbb{E}_X|\kappa_\mu(t)-\kappa_\mu(s)| \sum_{k=\kappa_\mu(s)}^{\kappa_\mu(t)-1} X_{k+1}^2 \mathbb E_\xi | \int_{\frac{k}\mu}^{\frac{k+1}\mu} \langle H_{X_{k+1}(s-\frac{k}\mu)}HA_{\frac{k}\mu,s}\phi,  g_\mu\left( T_{k,\mu}\right)\psi \rangle \mathrm ds|^2 \\
 \stackrel{\eqref{Jensen}}\le& \frac1\mu\mathbb{E}_X|\kappa_\mu(t)-\kappa_\mu(s)| \sum_{k=\kappa_\mu(s)}^{\kappa_\mu(t)-1} X_{k+1}^2 \int_{\frac{k}\mu}^{\frac{k+1}\mu} \mathbb E_\xi |\langle H_{X_{k+1}(s-\frac{k}\mu)}H A_{\frac{k}\mu,s} \phi,  g_\mu\left( T_{k,\mu}\right)\psi \rangle|^2 \mathrm ds \\
\le& \frac1\mu\mathbb{E}_X|\kappa_\mu(t)-\kappa_\mu(s)| \sum_{k=\kappa_\mu(s)}^{\kappa_\mu(t)-1} X_{k+1}^2 \int_{\frac{k}\mu}^{\frac{k+1}\mu} \mathbb E_\xi \|H_{X_{k+1}(s-\frac{k}\mu)}H A_{\frac{k}\mu,s} \phi\|^2 \|g_\mu\left( T_{k,\mu}\right)\psi\|^2 \mathrm ds
\end{split}
\end{equation}
where we used again Cauchy-Schwarz in the last line. Recalling definitions \eqref{e:partialflowA} and \eqref{b1}, we see that the two norm terms within each summand are independent with respect to $\mathbb P_\xi$, so we can separate them and then use the isometry \eqref{b5} and the unitarity of $H_t$:
\begin{equation}\label{b9}
\begin{split}
&\mathbb{E}|\langle\phi,I_1^\mu(t)-I_1^\mu(s)\rangle|^2 \\
 \le& \frac1\mu\mathbb{E}_X|\kappa_\mu(t)-\kappa_\mu(s)| \sum_{k=\kappa_\mu(s)}^{\kappa_\mu(t)-1} X_{k+1}^2 \int_{\frac{k}\mu}^{\frac{k+1}\mu} \mathbb E_\xi \|H_{X_{k+1}(s-\frac{k}\mu)}H A_{\frac{k}\mu,s} \phi\|^2 \cdot \mathbb E_\xi \|g_\mu\left(T_{k,\mu}\right)\psi\|^2 \mathrm ds \\
 =& \frac1\mu\mathbb{E}_X|\kappa_\mu(t)-\kappa_\mu(s)| \sum_{k=\kappa_\mu(s)}^{\kappa_\mu(t)-1} X_{k+1}^2 \|\psi\|^2 \int_{\frac{k}\mu}^{\frac{k+1}\mu} \mathbb E_\xi \|H A_{\frac{k}\mu,s} \phi\|^2 \mathrm ds \\
\end{split}
\end{equation}
Due to definition \eqref{e:partialflowA} and the stationary increments of the Wiener process, the integral simplifies to
\begin{equation}\label{b9.5}
\begin{split}
 \int_{\frac{k}\mu}^{\frac{k+1}\mu} \mathbb E_\xi \|H A_{\frac{k}\mu,s} \phi\|^2 \mathrm ds
 =& \int_{0}^{\frac{1}\mu} \mathbb E \|H A_{0,s} \phi\|^2 \mathrm ds \\
 \stackrel{\eqref{Guenter}}\le& 2 \int_{0}^{\frac{1}\mu} \mathbb E \|H (A_{0,s}-1) \phi\|^2 \mathrm ds
 +2 \int_{0}^{\frac{1}\mu} \mathbb E \|H\phi\|^2 \mathrm ds \le\frac C\mu
\end{split}
\end{equation}
where we used in the last inequality  Condition \ref{Bed}(I), namely that $\mathbb E \|H (A_{0,s}-1) \phi\|^2$ is bounded since it converges to zero. Substituting into \eqref{b9} and using Lemma \ref{lemm:kappa} (a), we get equicontinuity of $I_1^\mu(t)$:
\begin{equation*}
 \mathbb{E}|\langle\phi,I_1^\mu(t)-I_1^\mu(s)\rangle|^2 \le C(\sqrt{|t-s|}+(t-s)^2)
\end{equation*}

Starting again like in \eqref{b7}, we proceed with
\begin{align*}
 &\mathbb{E}|\langle\phi,I_2^\mu(t)-I_2^\mu(s)\rangle|^2
 =\mathbb{E}_X\mathbb E_\xi|\sum_{k=\kappa_\mu(s)}^{\kappa_\mu(t)-1} \int_{\frac{k}\mu}^{\frac{k+1}\mu} \langle H_{X_{k+1}(s-\frac{k}\mu)} A_{\frac{k}\mu,s} A \phi,  g_\mu\left( T_{k,\mu}\right)\psi \rangle \mathrm d\xi_s|^2\,.
\end{align*}
Now we interpret the sum of the integrals as one single Itô integral, to which we apply the Itô isometry:
\begin{align*}
 &\mathbb{E}|\langle\phi,I_2^\mu(t)-I_2^\mu(s)\rangle|^2
 =\mathbb{E}_X\sum_{k=\kappa_\mu(s)}^{\kappa_\mu(t)-1}  \int_{\frac{k}\mu}^{\frac{k+1}\mu}\mathbb E_\xi |\langle H_{X_{k+1}(s-\frac{k}\mu)}  A_{\frac{k}\mu,s} A \phi,  g_\mu\left( T_{k,\mu}\right)\psi \rangle|^2 \mathrm ds\,.
\end{align*}
Similarly to our treatment of the $I_1$ term, we arrive at
\begin{equation*}
 \mathbb{E}|\langle\phi,I_2^\mu(t)-I_2^\mu(s)\rangle|^2
 \le \frac{C}\mu\mathbb E |\kappa_\mu(t)-\kappa_\mu(s)| \le C|t-s| .
\end{equation*}
In the second inequality, we have used the fact that $|\kappa_\mu(t)-\kappa_\mu(s)|$, the number of jumps during $[s,t]$, is Poisson distributed with mean value $\mu|t-s|$.

The last term in \eqref{b6} can be estimated in the same way. One ends up with
\begin{align*}
 \mathbb{E}|\langle\phi,I_3^\mu(t)-I_3^\mu(s)\rangle|^2
 \le&\frac{C}{\mu^2}\mathbb E (\kappa_\mu(t)-\kappa_\mu(s))^2 \le C|t-s| .
\end{align*}

Thus, equicontinuity of $(f_\mu)$ is proven and, along the way, one gets the estimate
\begin{equation}\label{gleichgradig}
 \mathbb E|\langle\phi, (g_\mu(t) - g_\mu(s))\psi\rangle|^2 \le C(\sqrt{|t-s|} + (t-s)^2) .
\end{equation}
Consequently, there exists a subsequence $(\mu_k^1)_{k\in\mathbb N}$ of $(\mu)_{\mu\in\mathbb N}$ such that $\mathbb{E}\overline{Y_1\zeta_1}\langle\phi_1,g_{\mu_k^1}(\cdot)\psi\rangle$ converges uniformly in t for $k\to\infty$ (recall definition \eqref{b4} and the choice \eqref{b5.5}), a subsequence $(\mu_k^2)$ of $(\mu_k^1)$ such that also
$\mathbb{E}\overline{Y_2\zeta_2}\langle\phi_2,g_{\mu_k^2}(\cdot)\psi\rangle$ converges and so on. For the diagonal sequence $(\mu_k^k)$ obviously $\mathbb{E}\overline{Y_j\zeta_j}\langle\phi_j,g_{\mu_k^k}(\cdot)\psi\rangle$ converges for all $j\in\mathbb N$ uniformly in t. From the isometry \eqref{b5} and the Cauchy-Schwarz inequality, one infers that
\begin{equation*}
 L^2(\Omega,\mathfrak{H})\ni\eta\mapsto\lim_{k\to\infty}\mathbb{E}\langle\eta,g_{n_k^k}(t)\psi\rangle\in\mathbb{R}
\end{equation*}
is a continuous functional. This implies, via the Riesz representation theorem, the existence of a process $\psi_t$ with
\begin{equation*}
 \lim_{k\to\infty}\mathbb{E}\overline{Y\zeta}\langle\phi,g_{\mu_k^k}(t)\psi\rangle = \mathbb{E}\overline{Y\zeta}\langle\phi,\psi_t\rangle
\end{equation*}
for all $Y, \zeta,\phi$. Since the limit was obtained by the Arzela-Ascoli theorem, it is uniform in $t\in[0,T]$ if $Y, \zeta,\phi$ are fixed.

In order to show that the constructed $\psi_t$ is actually the solution of \eqref{e:SSE}, the convergence of the summands in the weak form of \eqref{b3}, i.e. in
\begin{align}\label{schwach1}
 \mathbb{E}\overline{Y\zeta}\langle\phi,g_{\mu^k_k}(t)\psi\rangle =& \mathbb{E}\overline{Y\zeta}\langle\phi,\psi\rangle + \sum_{l=1}^3\mathbb{E}\overline{Y\zeta}\langle\phi,I_l^{\mu^k_k}(t)\rangle ,
\end{align}
to their counterparts in the weak version of \eqref{e:SSE}, i.e. in
\begin{equation}\label{schwach2}
 \mathbb E\overline{Y\zeta}\langle\phi,\psi_t\rangle
= \mathbb E\overline{Y\zeta}\langle\phi,\psi\rangle
-i \mathbb E\overline{Y\zeta}\int_0^t\langle H\phi,\psi_s\rangle \textnormal ds
+ \mathbb E\overline{Y\zeta}\int_0^t\langle A\phi,\psi_s\rangle \textnormal d\xi_s
-\frac12 \mathbb E\overline{Y\zeta}\int_0^t\langle A^2\phi,\psi_s\rangle \textnormal ds ,
\end{equation}
has to be checked. For simplicity, we will write $\mu$ instead of $\mu_k^k$, but recall that we are dealing with a subsequence.
The left-hand side was just settled, so we start with the $l=1$ term: Applying Cauchy-Schwarz inequality, we arrive at
\begin{equation}\label{b12}
\begin{split}
 &|\mathbb{E}\overline{Y\zeta}\langle\phi,I_1^{ \mu}(t)\rangle - \mathbb{E}\overline{Y\zeta}\int_0^t\langle iH\phi,\psi_s\rangle \mathrm ds|^2 \\
\le& C\mathbb E |\langle\phi,I_1^{ \mu}(t)\rangle - \int_0^t\langle iH\phi,\psi_s\rangle \mathrm ds|^2 \\
 \le& C\mathbb E |\sum_{k=0}^{\kappa_ \mu(t)-1}(X_{k+1}-1) \int_{\frac{k}{ \mu}}^{\frac{k+1}{ \mu}}  \langle\phi, A_{\frac{k}{ \mu},s} H_{X_{k+1}(s-\frac{k}{ \mu})}H g_{ \mu}\left( T_{k,\mu}\right)\psi \rangle \mathrm ds |^2 \\
 &+ C\mathbb E |\sum_{k=0}^{\kappa_ \mu(t)-1} \int_{\frac{k}{ \mu}}^{\frac{k+1}{ \mu}}  \langle H_{X_{k+1}(s-\frac{k}{ \mu})}H (A_{\frac{k}{ \mu},s}-1) \phi,  g_{ \mu}\left( T_{k,\mu}\right)\psi \rangle \mathrm ds |^2 \\
 &+ C\mathbb E |\sum_{k=0}^{\kappa_ \mu(t)-1} \int_{\frac{k}{ \mu}}^{\frac{k+1}{ \mu}}  \langle (H_{X_{k+1}(s-\frac{k}{ \mu})}-1)H \phi,  g_{ \mu}\left( T_{k,\mu}\right)\psi \rangle \mathrm ds |^2 \\
 &+ C\mathbb E | \int_{\frac{\kappa_\mu(t)}{\mu}}^{t}  \langle H \phi, g_{\mu}(T_{\lfloor\mu s\rfloor,\mu})\psi \rangle \mathrm ds |^2 \\
 &+ C\mathbb E | \int_0^t \langle H \phi, (g_{\mu}(T_{\lfloor\mu s\rfloor,\mu})-g_\mu(s))\psi \rangle \mathrm ds |^2 \\
 &+ C\mathbb E | \int_0^t \langle H \phi, g_\mu(s)\psi-\psi_s \rangle \mathrm ds |^2 \\
 :=& \sum_{l=1}^6 J_l^\mu(t) .
\end{split}
\end{equation}
In the second step, we have substituted definintion \eqref{b3}, expanded the difference in a telescopic sum and applied the inequality \eqref{Guenter}.

Again, all the summands have to go to 0 for $\mu\to\infty$.
We start by expanding the square in $J_1^\mu$:
\begin{align*}
 J_1^\mu(t) \mathbb E_X\sum_{j,k=0}^{\kappa_\mu(t)-1}(X_{j+1}-1)(X_{k+1}-1)
 \mathbb E_\xi &\bigg(\overline{\int_{\frac{j}{ \mu}}^{\frac{j+1}{ \mu}}  \langle\phi, A_{\frac{j}{ \mu},s} H_{X_{j+1}(s-\frac{j}{ \mu})}H g_{ \mu}\left( T_{j,\mu}\right)\psi \rangle \mathrm ds} \\
 &\times\int_{\frac{k}{ \mu}}^{\frac{k+1}{ \mu}}  \langle\phi, A_{\frac{k}{ \mu},s} H_{X_{k+1}(s-\frac{k}{ \mu})}H g_{ \mu}\left( T_{k,\mu}\right)\psi \rangle \mathrm ds \bigg)
\end{align*}
Estimating the integrals by $\frac C\mu^2$ as we did in \eqref{b9} and \eqref{b9.5} and writing the sum as a square again, we continue
\begin{equation}\label{b15}
\begin{split}
 J_1^\mu(t)
 \le& \frac C{\mu^2} \left(\mathbb E_X \sum_{k=0}^{\kappa_\mu(t)-1}(X_{k+1}-1)\right)^2
 = \frac C{\mu^2} (\mathbb E \dots 1_{\kappa_\mu(t)\le6\mu t} + \mathbb E \dots 1_{\kappa_\mu(t) > 6\mu t})^2 \\
 \le& \frac C{\mu^2} (\mathbb E \dots 1_{\kappa_\mu(t)\le6\mu t})^2 + (\mathbb E \dots 1_{\kappa_\mu(t) > 6\mu t})^2
\end{split}
\end{equation}
Taking into account the independence of the $X_k$,
\begin{equation*}
 \frac1{\mu^2}(\mathbb E \dots 1_{\kappa_\mu(t)\le6\mu t})^2
 \le \frac1{\mu^2}\mathbb E \sum_{k=0}^{6\mu t-1}(X_{k+1}-1)^2 \le \frac{6\mu t}{\mu^2} \to0 ,
\end{equation*}
and using Jensen's inequality for $\mathbb E$, \eqref{Guenter} for the sum and the crude estimate $X_k\le\mu t$ which, by definition \eqref{kappa}, holds true for $k\le\kappa_\mu(t)$ and, via \eqref{Guenter}, implies $(X_{k+1}-1)^2\le 2\mu^2 t^2+2$,
\begin{align*}
 \frac1{\mu^2}(\mathbb E \dots 1_{\kappa_\mu(t) > 6\mu t})^2
 \le \frac1{\mu^2}\mathbb E \kappa_\mu(t) \sum_{k=0}^{\kappa_\mu(t)-1}(X_{k+1}-1)^2 1_{\kappa_\mu(t) > 6\mu t}
 \le C\mathbb E \kappa_\mu(t)^2 1_{\kappa_\mu(t) > 6\mu t} ,
\end{align*}
so that Lemma \ref{lemm:kappa} completes the proof that \eqref{b15} and thus $J_1^\mu(t)$ goes to 0.

All the other summands in \eqref{b12} are to be treated analogously, and we shall only indicate the points at which new arguments are needed or the conditions enter:

One gets
\begin{equation}
 J_2^\mu(t)\le C \sup_{0\le t\le\frac1\mu}\mathbb E\|H(A_{0,t}-1)\phi\|^2 ,
\end{equation}\label{b19}
which goes to 0 by condition (I),
and
\begin{equation}\label{b20}
 J_3^\mu(t)\le C\mathbb E \|(H_{\frac{X_{k+1}}{\mu}}-1)H\phi\|^2 .
\end{equation}
Since, by assumption (I), $\phi\in Dom(H^2)$, we can calculate
\begin{equation}\label{b21}
 \|(H_t-1)H\phi\| = \|\int_0^t H_sH^2\phi\text ds\| \le t\|H^2\phi\|
\end{equation}
and, substituting into \eqref{b20}, end up with
\begin{equation}\label{b22}
 J_3^\mu(t)\le \frac C{\mu^2}\mathbb E X_k^2 \to0 .
\end{equation}
Furthermore,
\begin{equation*}
 J_4^\mu(t)\le \frac C{\mu^2}\mathbb E(\kappa_\mu(t)-\mu t)^2 = \frac {C\mu t}{\mu^2} \to0 .
\end{equation*}
For $J_5^\mu$, we need the equicontinuity estimate \eqref{gleichgradig}:
\begin{equation}\label{b25}
\begin{split}
 J_5^\mu(t)\le& C\int_0^t\mathbb E (\sqrt{|T_{\lfloor\mu s\rfloor,\mu}-s|} + (T_{\lfloor\mu s\rfloor,\mu}-s)^2) \text ds \\
 \le& C \sum_{k=0}^{\mu t}\int_{\frac k\mu}^{\frac{k+1}\mu} \mathbb E(\sqrt{|T_{k,\mu}-s|} + (T_{k,\mu}-s)^2) 
\end{split}
\end{equation}
Now, within each integral,
\begin{equation*}
 \mathbb E(T_{k,\mu}-s)^2 = \mathbb V T_{k,\mu} + (s-\frac k\mu)^2 \le \frac k{\mu^2} + \frac1{\mu^2} \le\frac C\mu ,
\end{equation*}
\begin{equation*}
 \mathbb E\sqrt{|T_{k,\mu}-s|} \le \left(E(T_{k,\mu}-s)^2\right)^\frac14 \le C\mu^{-\frac14}
\end{equation*}
and, substituting back into \eqref{b25}, we end up with
\begin{equation*}
 J_5^\mu(t)\le C\mu^{-\frac14} \to0 .
\end{equation*}

$J_6^\mu(t)\to0$ by construction of $\psi_t$ as the uniform weak limit of $g_\mu(t)$.

For the $j=2$ summand in \eqref{schwach1},
\begin{equation*}
 |\mathbb{E}\overline{Y\zeta}\langle\phi,I_2^{n}(t)\rangle - \mathbb{E}\overline{Y\zeta}\int_0^t\langle A\phi,\psi_s\rangle \mathrm d\xi_s |^2
 \le C |\mathbb{E}\overline{\zeta}\int \dots \mathrm d\xi_s |^2 .
\end{equation*}
This can be reduced to an ordinary integral and treated like the preceding ones if one restricts w.l.o.g. $\zeta$ to the total set of iterated stochastic integrals \cite{Hackenbroch} defined via
\begin{equation}\label{itint}
M_0:\equiv1, M_{l+1}(t):=\int_0^ta(s)M_l(s)\textnormal d\xi_s
\end{equation}
with indicator functions a and applies the Itô isometry
\begin{equation*}
 \mathbb E (\int f\textnormal d\xi_t \int g\textnormal d\xi_t) = \mathbb E \int fg\textnormal dt .
\end{equation*}
The last summand of \eqref{schwach1} goes through like the first one. In place of \eqref{b19}, the terms
\begin{equation*}
 \mathbb E\|(A_{0,\frac1\mu}-1)A^j\phi\|^2
\end{equation*}
with $j=1\text{ and }2$ will come up, and one needs the condition $\phi\in Dom(A^4)$ in order to conclude from \eqref{Anfangswert} that they go to 0.

Consequently, the constructed process $\psi_t$ weakly solves equation \eqref{e:SSE}. As far as its solution is unique, it agrees with the one obtained by Holevo (\cite{Holevo}). We shall review his uniqueness proof in the appendix.
The same method also allows to start with any subsequence $g_{\tilde \mu}(t)\psi$ and extract a sub-subsequence converging weakly to the same $\psi_t$, so $g_{\mu}(t)\psi \rightharpoonup \psi_t$ holds true without restriction to a subsequence.

Compared to the right-hand side of \eqref{produkt}, the first factor is missing in the definition \eqref{b1} of $g_\mu(t)$. Of course, both versions have the same limit because one single factor approaches the identity for $\mu\to\infty$:
\begin{align*}
 \mathbb E Y\zeta\langle\phi, (H_{t-T_{\kappa_\mu(t),\mu}}\prod_{k=0}^{\kappa_\mu(t)-1} A_{\frac{k}{\mu},\frac{k+1}{\mu}} H_{\frac{X_{k+1}}\mu} - g_\mu(t)) \psi\rangle
 \le& C\mathbb E\|(H_{t-T_{\kappa_\mu(t),\mu}} - 1)\phi\|^2 \\
 \stackrel{\eqref{b21}}\le& C \mathbb E(t-T_{\kappa_\mu(t),\mu})^2 = \frac C{\mu^2} \to0
\end{align*}
where we took into account that $t-T_{\kappa_\mu(t),\mu}$, the waiting time for the last jump of $\kappa_\mu$ before t, has exponential distribution with parameter $\mu$.

Since $\lim_{n\to\infty}\mathbb{E}\|g_{\mu}(t)\psi\|²=\|\psi\|²$ holds true and, under the conditions from \cite{Mora},  $\psi\in\mathcal{M}\cap D(R)$ also satisfies $\mathbb{E}\|\psi_t\|^2=\|\psi\|^2$, in this case the convergence $g_{\mu}(t)\to C_{0,t}$ is even strong. Since $C_{0,t}$ as well as all operators $g_{\mu}$ have $L^2(\Omega,\mathfrak{H})$-norm 1, the strong convergence on the dense subset $\mathcal M\cap D(R)$ implies the convergence on the whole of $\mathfrak{H}$.
\end{proof}

\begin{remark}
 It is not essential whether we group the factors in the order $A_{\frac k\mu, \frac{k+1}\mu}H_{\frac{X_{k+1}}{\mu}}$, as we did in the telescopic sum expansion \eqref{b2}, or in the order $H_{\frac{X_{k+1}}{\mu}}A_{\frac k\mu, \frac{k+1}\mu}$. If we had decided to do the latter, we would have had to require
\begin{equation*}
 \|A^j(H_t-1)\phi\|\le Ct \text{ for } j\in\{1,2\}
\end{equation*}
instead of
\begin{equation*}
 \lim_{t\to0} \mathbb E\|H(A_{0,t}-1)\phi\|^2 =0
\end{equation*}
in condition (I). In both cases, one can get the product formula in the form \eqref{produkt} or with $A_{\frac k\mu, \frac{k+1}\mu}$ coming before $H_{\frac{X_{k+1}}{\mu}}$.
\end{remark}

I still remains to provide the following technical lemma.
\begin{lemma}\label{lemm:kappa}
As in Theorem \ref{thm:Produkt} let $X_1, X_2, \dots$ denote  independent exponentially distributed random variables and let $\kappa_\mu(t)$ be the ``number of hittings''
\begin{equation}\label{kappa}
 \kappa_\mu(t) := \max\biggl\{k\in\mathbb{N}:\sum_{j=1}^k \frac{X_i}\mu\le t \biggr\}
\end{equation}
before t. Then
\begin{align*}
 a)& \frac{1}{\mu^2}\mathbb E\biggl[ |\kappa_\mu(t)-\kappa_\mu(s)| \sum_{k=\kappa_\mu(s)}^{\kappa_\mu(t)-1}X_{k+1}^2\biggr] \le C(\sqrt{|t-s|} + |t-s|^2) \textnormal{ } uniformly\textnormal{ } in \textnormal{ } \mu \in \mathbb{N} \\
 b)& \lim_{\mu\to\infty} \mathbb E \chi_{\kappa_\mu(t)\ge 6\mu t} \kappa_\mu(t)^2 = 0 \textnormal{ }for\textnormal{ }every\textnormal{ }t .
\end{align*}
\end{lemma}

\begin{proof}
 We start with the case $s=0$, decompose
\begin{equation}\label{l1}
 \frac{1}{\mu^2}\mathbb E \biggl[|\kappa_\mu(t)| \sum_{k=0}^{\kappa_\mu(t)-1}X_{k+1}^2\biggr]
 = \frac{1}{\mu^2}\mathbb E\bigl[ \dots 1_{\kappa_\mu(t)\le 6\mu t}\bigr] + \frac{1}{\mu^2}\mathbb E\bigl[ \dots 1_{\kappa_\mu(t)> 6\mu t}\bigr]
\end{equation}
(you will see where the ``6'' comes from) and find
\begin{equation*}
 \frac{1}{\mu^2}\mathbb E \dots 1_{\kappa_\mu(t)\le 6\mu t}
 \le \frac{1}{\mu^2} \sum_{k=0}^{6\mu t-1}\mathbb E 6\mu tX_{k+1}^2 \le Ct^2 .
\end{equation*}
As for the second summand in \eqref{l1}, we know from the definition of $\kappa_\mu$ that at least all $X_k\le \mu t$, so
\begin{equation}\label{l2}
 \frac{1}{\mu^2}\mathbb E \dots 1_{\kappa_\mu(t)> 6\mu t}
 \le t^2 \mathbb E |\kappa_\mu(t)|^21_{\kappa_\mu(t)> 6\mu t} .
\end{equation}
$\kappa_\mu(t)$ is Poisson distributed with parameter $\mu t$ and the important point which we must not hide by too crude estimates is the exponential decay of
$\mathbb P(\kappa_\mu(t)> 6\mu t)$ in $\mu t$, therefore we compute
\begin{align*}
 &\frac{1}{\mu^2}\mathbb E \dots 1_{\kappa_\mu(t)> 6\mu t}
 = Ct^2\sum_{k=\lceil 6\mu t\rceil}^\infty k^2\frac{(\mu t)^k}{k!} e^{-\mu t} \\
 =& Ct^2 e^{-\mu t} \mu t \sum_{k=\lceil 6\mu t\rceil}^\infty k\frac{(\mu t)^{k-1}}{(k-1)!}
 = Ct^2 e^{-\mu t} \mu t \left( \sum_{k=(\lceil 6\mu t\rceil-1)\vee0}^\infty k\frac{(\mu t)^k}{k!} + \sum_{k=(\lceil 6\mu t\rceil-1)\vee0}^\infty \frac{(\mu t)^k}{k!} \right) \\
 =& Ct^2 e^{-\mu t} \mu t \left( \mu t\sum_{k=(\lceil 6\mu t\rceil-2)\vee0}^\infty \frac{(\mu t)^k}{k!} + \sum_{k=(\lceil 6\mu t\rceil-1)\vee0}^\infty \frac{(\mu t)^k}{k!} \right) .
\end{align*}
Writing
$$\{6\mu t\}:=(\lceil 6\mu t\rceil-1)\vee 0 ,$$
applying the Stirling formula and observing that
\begin{equation}\label{l5}
\frac{\mu t}{\{6\mu t\}}\le \frac13 \text{ if } \{6\mu t\}\ge1 ,
\end{equation}
we obtain
\begin{equation}\label{l10}
\begin{split}
 \frac{1}{\mu^2}\mathbb E \dots 1_{\kappa_\mu(t)> 6\mu t}
 \le& Ct^2 e^{-\mu t} ((\mu t)^2+\mu t) \sum_{k=\{6\mu t\}}^\infty \frac{(\mu t)^k}{k!} \\
 \le&
 \begin{cases}
  Ct^2 e^{-\mu t} ((\mu t)^2+\mu t) \frac{(\mu t)^{\{6\mu t\}}}{\{6\mu t\}!} \sum_{k=0}^\infty \left(\frac{\mu t}{\{6\mu t\}}\right)^k & \textnormal{for } \{6\mu t\}\ge1 \\
  Ct^2((\mu t)^2+\mu t) & \textnormal{for } \{6\mu t\}=0
 \end{cases} \\
 \stackrel{\eqref{l5}}\le&
 \begin{cases}
  Ct^2 e^{-\mu t} ((\mu t)^2+\mu t) \frac{(\mu t)^{\{6\mu t\}}e^{\{6\mu t\}}}{\sqrt{2\pi\{6\mu t\}}\{6\mu t\}^{\{6\mu t\}}} \frac{1}{1-\frac13} & \textnormal{for } \{6\mu t\}\ge1 \\
  Ct^2((\mu t)^2+\mu t) & \textnormal{for } \{6\mu t\}=0
 \end{cases} \\
 \le&
 \begin{cases}
  Ct^2 ((\mu t)^2+\mu t) \left(\frac e3\right)^{\{6\mu t\}}  & \textnormal{for } \{6\mu t\}\ge1 \\
  Ct^2((\mu t)^2+\mu t) & \textnormal{for } \{6\mu t\}=0
 \end{cases} \\
 \le& Ct^2 (\{6\mu t\}^2+\{6\mu t\}+1)\left(\frac e3\right)^{\{6\mu t\}} \le Ct^2 .
\end{split}
\end{equation}
This proves part a) for $s=0$; note that, in going from \eqref{l2} to the second last step in \eqref{l10}, we have also proven part b) along the way.

In case that $s\ne0$, one would be inclined to say that
\begin{equation*}
 \frac{1}{\mu^2}\mathbb E\biggl[ |\kappa_\mu(t)-\kappa_\mu(s)| \sum_{k=\kappa_\mu(s)}^{\kappa_\mu(t)-1}X_{k+1}^2\biggr]
 = \frac{1}{\mu^2}\mathbb E\biggl[ \kappa_\mu(|t-s|) \sum_{k=0}^{\kappa_\mu(|t-s|)-1}X_{k+1}^2\biggr]
\end{equation*}
because the Poisson process $\kappa_\mu$ has stationary increments. In doing so, one would, however, neglect that the ``first'' X summand decomposes into
\begin{equation*}
 X_{\kappa_\mu (s)+1}=(\sum_{k=0}^{\kappa_\mu (s)}X_{k+1} - \mu s) + (\mu s - \sum_{k=0}^{\kappa_\mu (s)-1}X_{k+1}) =: Y+Z ,
\end{equation*}
where $\frac Y\mu$ is the waiting time for the last jump of $\kappa_\mu$ before s and $\frac Z\mu$ the waiting time for the next jump after s. A correct use of stationarity implies
\begin{equation}\label{l20}
 \frac{1}{\mu ^2}\mathbb E |\kappa_\mu (t)-\kappa_\mu (s)| (Z^2+\sum_{k=\kappa_\mu (s)+1}^{\kappa_\mu (t)-1}X_{k+1}^2)
 = \frac{1}{\mu ^2}\mathbb E |\kappa_\mu (|t-s|)| \sum_{k=0}^{\kappa_\mu (|t-s|)-1}X_{k+1}^2 \le C|t-s|^2 .
\end{equation}
Taking into account that Y and Z have independent exponential distributions and that \\
$X_{\kappa_\mu (s)+1}^2 \le 2Y^2+2Z^2$, we end up with
\begin{align*}
 &\frac{1}{\mu ^2}\mathbb E |\kappa_\mu (t)-\kappa_\mu (s)| \sum_{k=\kappa_\mu (s)}^{\kappa_\mu (t)-1}X_{k+1}^2 \\
 \le& \frac{1}{\mu ^2}\mathbb E |\kappa_\mu (t)-\kappa_\mu (s)| (2Y^2+Z^2)
 + \frac{1}{\mu ^2}\mathbb E |\kappa_\mu (t)-\kappa_\mu (s)| (Z^2+\sum_{k=\kappa_\mu (s)+1}^{\kappa_\mu (t)-1}X_{k+1}^2) \\
 \le& \frac{1}{\mu ^2}(\mathbb E |\kappa_\mu (t)-\kappa_\mu (s)|^2)^{\frac 12} (\mathbb E (2Y^2+Z^2)^2)^\frac12 + C|t-s|^2
 \le C(\sqrt{|t-s|} + |t-s|^2)
\end{align*}
where, in the second last step, we have used Cauchy-Schwarz and \eqref{l20}.
\end{proof}

\section*{Appendix: Uniqueness of the weak solution of \eqref{e:SSE}}

We recall Holevo's proof (p. 487 of \cite{Holevo}) that weak solutions of dissipative stochastic equations are unique under the growth condition
\begin{equation}\label{Wachstum}
\mathbb{E}\|\psi_t\|^2 \le \|\psi_0\|^2 e^{ct} \textnormal{ for some } c ,
\end{equation}
provided that
$$ -iH + \frac12 A^2 \text{ is maximally accretive},$$
and explain one essential step in the his proof which remained unclear to us in \cite{Holevo} (see below). For simplicity, we formulate things in terms of our equation \eqref{e:SSE} instead of his more general framework.

Suppose that $\psi^1_t$ and $\psi^2_t$ both solve \eqref{e:SSE} with the same initial condition $\psi_0$ and satisfy \eqref{Wachstum}. Defining $\delta_t:=\psi^1_t-\psi^2_t$, we have to show $\mathbb{E}\overline{M_l(t)}\langle\phi,\delta_t\rangle=0$ for all $\phi$ in a dense subset of $\mathfrak H$ and iterated integrals $M_l$ as specified in \eqref{itint}. The weak form \eqref{schwach2} of equation \eqref{e:SSE} gives
\begin{equation}\label{733}
 \overline{M_0(t)}\langle\phi,\delta_t\rangle=\langle\phi,\delta_t\rangle
 = \int_0^t\langle A\phi,\delta_s\rangle\mathrm{d}\xi_s
 +\int_0^t\langle (iH-\frac{1}{2}A^2)\phi,\delta_s\rangle\mathrm{d}s
\end{equation}
and, according to the product formula for It\^o differentials,
\begin{equation}\label{734}
\begin{split}
 \overline{M_l(t)}\langle\phi,\delta_t\rangle =& \int_0^t \overline{M_{l-1}(s)a(s)}\langle\phi,\delta_s\rangle \textnormal d\xi_s
 +\int_0^t \overline{M_l(s)}\langle A\phi,\delta_s\rangle \textnormal d\xi_s
 +\int_0^t \overline{M_l(s)}\langle (iH-\frac{1}{2}A^2)\phi,\delta_s\rangle \textnormal ds \\
 &+\int_0^t \overline{M_{l-1}(s)} \overline{a(s)}\langle A\phi,\delta_s\rangle \textnormal{d}s
\end{split}
\end{equation}
for $l\ge1$. Taking for granted for a moment that the stochastic integrals have mean value 0 and defining $m_l(t):=\mathbb{E}\overline{M_l(t)}\delta_t$, these equations imply
\begin{equation}\label{735}
 \langle\phi,m_0(t)\rangle=\int_0^t\langle (iH-\frac{1}{2}A^2)\phi,m_0(s)\rangle\textnormal ds
\end{equation}
and for $l\ge1$
\begin{equation}\label{736}
 \langle\phi,m_l(t)\rangle = \int_0^t\langle a(s)A\phi,m_{l-1}(s)\rangle \textnormal ds
 +\int_0^t\langle (iH-\frac{1}{2}A^2)\phi,m_l(s)\rangle \textnormal ds .
\end{equation}
and one can inductively prove $m_l\equiv0$: Starting with $l=0$, we somehow have to bring the two terms $m_0$ in equation \eqref{735} together. For this sake, we take the Laplace transform, perform partial integration on the right-hand side und thus get rid of the additional integral in front of one $m_0$:
\begin{align*}
 \int_0^te^{-\lambda s}\langle\phi,m_0(s)\rangle\textnormal ds
 =& \int_0^te^{-\lambda s}\int_0^s\langle (iH-\frac{1}{2}A^2)\phi,m_0(r)\rangle \textnormal dr\textnormal ds \\
 =& \frac{1}{\lambda}\int_0^te^{-\lambda s}\langle (iH-\frac{1}{2}A^2)\phi,m_0(s)\rangle \textnormal ds
\end{align*}
Multiplying with $\lambda$ and dragging out the scalar products, we get
\begin{equation}\label{290}
 \langle(iH-\frac{1}{2}A^2-\lambda)\phi,\int_0^te^{-\lambda s}m_0(s)\textnormal ds\rangle = 0 .
\end{equation}
Now let $\lambda>0$. Since $-iH+\frac{1}{2}A^2$ was assumed to be m-accretive, $iH-\frac{1}{2}A^2-\lambda$ is surjective and \eqref{290} implies $\int_0^te^{-\lambda s}m_0(s)\textnormal ds=0$
for all $\lambda>0$ and thus $m_0\equiv0$.
In the induction step $l-1\to l$, the first integral in \eqref{736} is 0 according to the induction hypothesis, so $m_l\equiv0$ follows from this equation along the same lines as $m_0\equiv0$.

We still have to prove that the sum of the stochastic integrals in \eqref{734}, i.e. an expression of the form
\begin{equation*}
 J_t:=\int_0^t M_{l-1}(s)X_s \textnormal d\xi_s
 +\int_0^t M_l(s)Y_s \textnormal d\xi_s
\end{equation*}
with
\begin{equation}\label{a2}
\sup_{0\le s\le t}\mathbb{E}X_s^2, \sup_{0\le s\le t}\mathbb{E}Y_s^2<\infty
\end{equation}
according to \eqref{Wachstum}, has mean value 0. Note that the integrands are products of complex $L^2(\Omega)$-processes, so the usual $L^2$-theory of stochastic integrals, providing their martingale property, does not apply. Holevo claims without proof that already
\begin{equation*}
\sup_{0\le s\le t}\mathbb{E}|M_{l-1}(s)X_s|, \sup_{0\le s\le t}\mathbb{E}|M_l(s)Y_s|<\infty
\end{equation*}
implies the martingale property of $J_t$, but we do not see the reason. 

Our argument is as follows:
We compute $\mathbb EJ_t$, making use of the additional facts that $J_t \in L^1(\Omega)$ (because all the other summands in \eqref{734} are) and that the $M_l$ are much smoother than typical $L^2$-processes:

Define the stopping times
\begin{equation*}
\sigma_n:=\inf\{t\ge0:|M_{l-1}(t)|\ge n\textnormal{ or }|M_l(t)|\ge n\} .
\end{equation*}
Then
\begin{equation*}
 \begin{split}
  J_t=& \int_0^t M_{l-1}^{\sigma_n}(s)X_s \textnormal d\xi_s +\int_0^t M_l^{\sigma_n}(s)Y_s \textnormal d\xi_s \\
  &+\int_0^t (M_{l-1}(s)-M_{l-1}^{\sigma_n}(s))X_s \textnormal d\xi_s +\int_0^t (M_l(s)-M_l^{\sigma_n}(s))Y_s \textnormal d\xi_s .
 \end{split}
\end{equation*}
The first two integrands are majorized by $n|X_s|$ and $n|Y_s|$, so, in view of \eqref{l2}, the integrals have mean value zero. The last two integrands are zero before being stopped, in particular on $\{\sigma_n> t\}$, so the Riemann sums approximating the integrals in probability and consequently the integrals themselves are zero on this set. Summarizing,
\begin{equation*}
 \begin{split}
  \mathbb EJ_t=& \mathbb E \chi_{\{\sigma_n\le t\}}\left( \int_0^t (M_{l-1}(s)-M_{l-1}^{\sigma_n}(s))X_s \textnormal d\xi_s +\int_0^t
  (M_l(s)-M_l^{\sigma_n}(s))Y_s \textnormal d\xi_s\right) \\
  =& \mathbb E \chi_{\{\sigma_n\le t\}} J_t - \mathbb E \chi_{\{\sigma_n\le t\}} \int_0^t M_{l-1}^{\sigma_n}(s)X_s \textnormal d\xi_s
  - \mathbb E \chi_{\{\sigma_n\le t\}} \int_0^t M_{l}^{\sigma_n}(s)Y_s \textnormal d\xi_s
 \end{split}
\end{equation*}
Now all we need is a sufficient decay of $P(\{\tau_n\le t\})$ for $n\to\infty$, namely
\begin{equation}\label{Menge}
 \mathbb P(\{\sigma_n\le t\})=\mathcal O\left(\frac{1}{n^3}\right) .
\end{equation}
Then $\mathbb E \chi_{\{\sigma_n\le t\}} J_t \to0$ for $n\to\infty$ by the dominated convergence theorem and Cauchy-Schwarz inequality and It\^o's formula would give us
\begin{align*}
 |\mathbb E \chi_{\{\sigma_n\le t\}} \int_0^t M_{l-1}^{\sigma_n}(s)X_s \textnormal d\xi_s|^2
 \le \mathbb P(\{\sigma_n\le t\})^2 \mathbb E |\int_0^t M_{l-1}^{\sigma_n}(s)X_s \textnormal d\xi_s|^2
 \le \frac{C}{n^3}n^2 \int_0^t \mathbb{E}|X_s|^2 \textnormal ds \to0
\end{align*}
and the same for the last summand.
According to Doob's inequality,
\begin{equation*}
 \mathbb P(\{\sigma_0\le t\}) = \mathbb P(\{\sup_{0\le s\le t}|\xi_s|\ge n\}) \le \frac{\mathbb E|\xi_t|^3}{n^3} \textnormal{ and}
\end{equation*}
\begin{align*}
 \mathbb P(\{\sigma_n\le t\}) \le \mathbb P(\{\sup_{0\le s\le t}|M_{l-1}(s)|\ge n\}) 
 +\mathbb P(\{\sup_{0\le s\le t}|M_{l}(s)|\ge n\})
 \le \frac{\mathbb E|M_{l-1}(t)|^3 + \mathbb E|M_l(t)|^3}{n^3}
\end{align*}
so that \eqref{Menge} is proven once we know $\mathbb E|M_l(t)|^3<\infty$ for all l. We prove this inductively, starting from
$\mathbb E|M_1(t)|^3 = \mathbb E|\xi_t|^3<\infty$ and using the Burkholder inequality (e.g. theorem 73 in \cite{Protter}) in order to write $\mathbb E|M_{l}(t)|^3$ in terms of $\mathbb E|M_{l-1}(t)|^3$:
\begin{align*}
 \mathbb E|M_l(t)|^3 =& \mathbb E|\int_0^t a_sM_{l-1}(s)\textnormal d\xi_s|^3
 \le C \mathbb E [\int_0^t a_sM_{l-1}(s)\textnormal d\xi_s, \int_0^t a_sM_{l-1}(s)\textnormal d\xi_s]^\frac{3}{2} \\
 =& C \mathbb E \left(\int_0^t a_s^2M_{l-1}(s)^2\textnormal ds\right)^\frac{3}{2}
 \overset{(1)}{\le} C\sqrt{t} \mathbb E \int_0^t a_s^3M_{l-1}(s)^3\textnormal ds
 \le Ct^\frac{3}{2} \mathbb E \sup_{0\le s\le t}|M_{l-1}(s)|^3 \\
 \overset{(2)}{\le}& 8Ct^\frac{3}{2} \mathbb E |M_{l-1}(t)|^3
\end{align*}
$[\cdot,\cdot]$ denotes the quadratic variation and can be computed according to
\begin{equation*}
[\int_0^tX_s\textnormal d\xi_s, \int_0^tX_s\textnormal d\xi_s]=\int_0^tX_s^2\textnormal ds ,
\end{equation*}
in (1) we used Jensen's inequality and in (2) again Doob's inequality.

\bibliographystyle{plain}

\def\cprime{$'$} \def\cprime{$'$}

\end{document}